\documentclass[twocolumn,10pt]{asme2ej}

\usepackage{epsfig}
\newtheorem{thm}{Theorem}[section]
\newtheorem{cor}{Corollary}[section]
\newtheorem{rem}{Remark}[section]
\title{On the Post-Peak Structural Response due to Softening with Localization}

\author{Hui-Hui Dai \thanks{Address all correspondence to this author.}
    \affiliation{
    Department of Mathematics\\
   and Liu Bie Ju Centre  \\
   for Mathematical Sciences,  \\
   City University of Hong Kong,\\
    83 TatChee Avenue,\\
     Kowloon Tong, Hong Kong\\
    Email: mahhdai@cityu.edu.hk\\Tel:
+852 27888660; fax: +852 27888561.
    }
}

\author{Xiaowu Zhu \\
    \affiliation{ School of Mathematics and Statistics,\\
     Wuhan University, Wuhan 430072, P.R. China\\Department of Mathematics,\\
      City University of Hong Kong,\\
       83 TatChee Avenue,\\
        Kowloon Tong, Hong Kong
    }
}

\author{Zhen Chen\\
    \affiliation{Department of Civil and Environmental Engineering, \\
    University of Missouri-Columbia, Columbia, MO 65211-2200 USA\\
        Department of Engineering Mechanics, \\
        Dalian University of Technology, \\
         Dalian 116024, P.R. China.

    }
}

\begin{document}

\maketitle

\begin{abstract}
{\it An analytical study is taken to investigate the relationship
between material softening and structural softening through the use
of a model problem in one dimension. With general nonlinear
assumptions on the constitutive relations, it turns out that the
governing equations can be viewed as a system of parametric
equations, which couple the size effect and the nonlinear effect.
Compared with the bilinear assumptions in previous literature, we
find that the nonlinear assumptions herein capture more details in
the post-peak structural response. After doing standard mathematical
analysis to the nonlinear equations, we manage to derive necessary
and sufficient conditions for the occurrence of four important
post-peak cases, which are often observed in experiments. In
particular, our analysis reveals that the mechanism of the
snap-through phenomenon is due to the convexity change of the
constitutive curve of the softening part. Mathematical examples are
also given to illustrate the proposed procedures.}
\end{abstract}

%

\section{Introduction}
Strain-softening, i.e., the decrease of stress with the increase of
strain, is such a common phenomenon that has been recorded for a
variety of materials, like concrete, rocks, ceramics, metals, etc.
Bazant et al., \cite{Bazant1984} gave a comprehensive review of this
phenomenon and analyzed its mechanism from a continuum point of
view. Moreover, it is well-known that strain softening is always
accompanied by highly localized deformations of the specimen
(\cite{van1986,Read1984}). Due to the importance of softening
phenomenon in structural safety assessment, many efforts have been
made in the past decades to investigate strain-softening with
localization experimentally, numerically, and analytically, as
reviewed by \cite{Bazant1997,Labuz2007}.

Snap-back may be one of the most interesting and perhaps most common
structural instability phenomena observed in experiments. It shows
that the load-displacement curve displays a positive slope after
attaining the peak load. de Borst \cite{Borst1987} demonstrated the
possibility of snap-back behavior on structural level by means of
two concrete structures: a reinforced concrete and an unreinforced
specimen. In order to simulate the highly localized failure mode in
a strain-softening solid, a modified arc-length control method was
used in that paper. Later, Rots and de Borst \cite{Rots1989} did a
tensile test on concrete specimens and analyzed it by using the
finite element method, with a particular attention on the snap-back
behavior. He et al., \cite{He1990} studied the class \textrm{II}
behavior (snap-back) of rock with a spring model, which was
characterized by non-uniform failure. Unloading-reloading tests were
also conducted in the post failure region in that paper. One of
their results is that, if inelastic strain increases slower than the
elastic strain decreases, rock shows class \textrm{II} behavior.

Jansen et al.\cite{Jansen1997} did an experiment on concrete
cylinders by using the feedback-control method. From two test
series, the stress-displacement behavior for different
height-diameter ratios with normal strength and high strength were
obtained. They found that the pre-peak segment of the
stress-displacement curves agrees well with the pre-peak part of the
stress-strain curves, while the post-peak segment shows a strong
dependence on the geometric size, namely the radius-length ratio.
More specifically, the longer the specimen is, the steeper the
post-peak segment of the stress-displacement curves becomes. The
feedback-control method was also used in Subramaniam et al.
\cite{Subramaniam1998} to test concrete in torsion, and snap-back
was also found in the experiment.

Some analytical studies were also taken to investigate softening
with localization. With the use of a one-dimensional model, Schreyer
and Chen \cite{Schreyer1986} analyzed the snap-back phenomenon and
found the important size effect on the instability. Due to the
simplicity of bilinear assumptions on the constitutive relations,
further features like snap-through were lost in the result, although
in some experiments this feature was observed (see van Vilet and van
Mier \cite{van2000}). The same constitutive relations were also
assumed in Chen et al. \cite{Chen2008} to analyze the stability in
some hierarchical structures. In a more complex setting with certain
nonlinear assumptions on the constitutive relations, Sundara Raja
Iyengar et al. \cite{Sundara2002} took an analytical study. By using
the fictitious crack model (FCM) developed by Hillerborg, they found
the effect of the softening exponent $n$ on the size effect and
snap-back behavior of beams, while the stress-displacement relation
was assumed as a general power law function. Dai et
al.\cite{Dai2008} constructed the analytical solutions for
localizations in a hyperelastic slender cylinder. With the use of
coupled series-asymptotic expansions approach and phase plane
analysis, they solved the partial differential equations and found
that the width of the localization zone depends on the material
parameters in the post-peak region. Further, they showed that there
is a snap-back phenomenon when the radius-length ratio is relatively
small, which agrees well with experimental observations. Dai et
al.\cite{Dai2009} showed a similar result for hyperelastic shape
memory alloys. Gradient theory may be another powerful tool in
dealing with localization of deformation (see Triantafllidis and
Aifantis \cite{Tri1986}). For example, Triantafllidis and
Bardenhagen \cite{Tri1993} investigated the issues of instability
and imperfection sensitivity of the solutions of a boundary value
problem in one dimension. Their results also revealed some important
size effect.

To the authors' knowledge, however, there is not any analytical
study with general nonlinear constitutive relations in the open
literature which explores the role played by the convexity of the
constitutive curve of the softening part and the coupling effect
between this convexity and the size. Also, both snap-back and
snap-through were observed in some experiments, but no analytical
results are available for explaining the transition from snap-back
to snap-through. We shall explore these aspects in this paper. To
gain insight into the post-peak response, we study the same
one-dimensional structure as considered in
\cite{Schreyer1986,Willam1985,Crisfield1982,Schreyer1984}. The
difference is that here we use general nonlinear constitutive
relations, instead of the bilinear ones used in these papers. First,
we set up the stress-strain equations for the structure in the
post-peak region, which are nonlinear as compared with the bilinear
case. After some analysis, we derive the mathematical conditions for
the occurrence of several important curves as frequently observed in
experiments, including the snap-through (which cannot be captured by
the bilinear assumptions). Finally, an example is given to
illustrate these cases, and the post-peak curves are consistent with
our theoretical predictions.
\section{Model Problem}
To simulate post-peak experiments, we consider a structure with a
serial arrangement of intact elastic and strain-softening zones.
This model was used by several researchers, such as
\cite{Willam1985,Crisfield1982,Schreyer1984} in the early years. In
\cite{Schreyer1986}, it was introduced to analyze strain-softening
with bilinear assumptions on the constitutive relations. As shown in
Figure 1(a), the structure is a bar of length $L=a+b$ with a unit
cross-sectional area. That is to say, it is composed of two segments
(segment A with length $a$ and segment B with length $b$). The two
segments are usually described by similar constitutive equations,
and the main difference is that the limit stress for B is slightly
less than that of A. Therefore, if the stress on the structure is
such that the strain in region B exceeds the value at the limit
state, then softening will occur. It is assumed that softening
occurs uniformly over a localized region B under quasi-static
loading.

In order to consider a general nonlinear case, the constitutive
relations for the two regions are set as: the loading and unloading
segments of region A are two arbitrary functions $f_{11}$ and
$f_{12}$ respectively, while the loading and softening segments of
region B are two arbitrary functions $f_{21}$ and $f_{22}$
respectively. Moreover, we assume that the foregoing nonlinear
functions are twice differentiable with $f^{'}_{11}>0$,
$f^{'}_{12}>0$, $f^{'}_{21}>0$, $f^{'}_{22}<0$. The limit stress for
region A, denoted by $\sigma_a$, is assumed to be slightly larger
than that of region B, which is denoted by $\sigma_0$. The details
are shown in Figure 1(b) (where $f_1$ is used to denote both the
pre-peak and the post-peak segments, as region A only experiences
loading or unloading ). As the post-peak curve of the structure is
our main concern, in the following derivation, for simplicity, we
use $f_1$, $f_2$ to denote the post-peak curves of region A and
region B respectively, unless otherwise specified.
\begin{figure}
   \centering
  \begin{minipage}{0.45\textwidth}
     \centering \includegraphics[width=70mm,height=60mm]{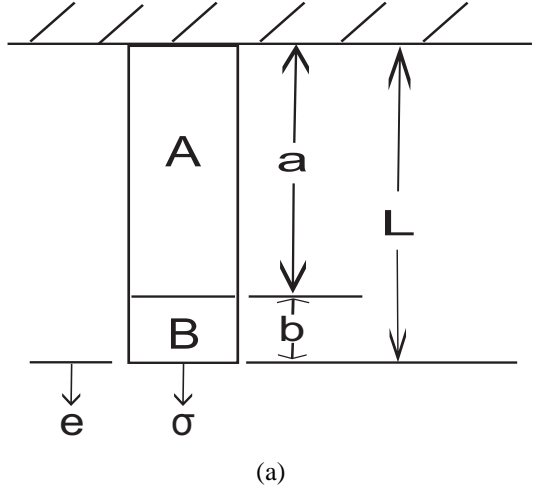}\\(a)
  \end{minipage}
  \begin{minipage}{0.45\textwidth}
     \centering \includegraphics[width=70mm,height=65mm]{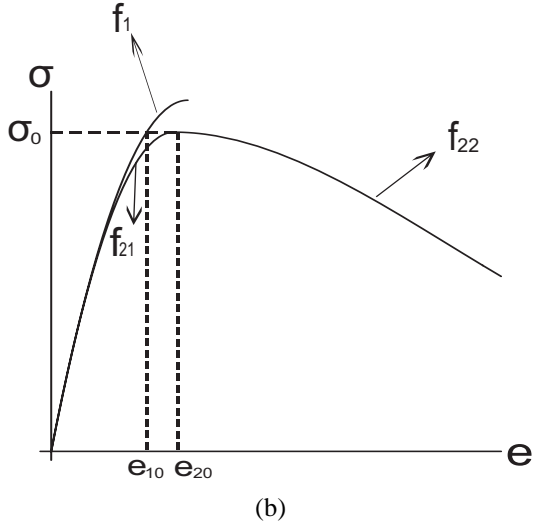}\\(b)
  \end{minipage}
  \caption{(a) One-dimensional model problem; (b)
Stress-strain relations for A and B.}
\end{figure}

As to the post-peak response, for a strain softening material with a
serial setting (cf. Figure 1(a)), region A is in an unloading
process and region B experiences strain softening. Given the values
of strain in regions A and B, say $e_1$ and $e_2$, respectively,
then the composite strain for the complete structure is given by
\begin{equation}
e=\frac{ae_1+be_2}{L}=(1-n)e_1+ne_2,
\end{equation}
where $n=b/L$. Since we consider it as a quasi-static problem, the
composite stress is then given by
\begin{equation}
\sigma=f_1(e_1)=f_2(e_2).
\end{equation}
In fact, one can easily see that (1) and (2) are also true if $f_2$
is used to denote both the pre-peak and post-peak segments. Here,
for the post-peak region, we consider only when $\sigma
\geq\sigma^\star$ ($\sigma^\star$ represents the lowest stress value
at which the bar breaks), and denote $e_{11}$ and $e_{21}$ the
values such that $f_1(e_{11})=f_2(e_{21})=\sigma^\star$. Then, for
the post-peak region, we have $e_1\in [e_{11},e_{10}]$ and
$e_2\in[e_{20},e_{21}]$ (see Figure 1(b) for the definitions of
$e_{10}$ and $e_{20}$). From equation (2), we get
$e_2=f_2^{-1}[f_1(e_1)]$ (or $e_1=f_1^{-1}[f_2(e_2)]$). Thus (1) and
(2) can be transformed into the system
\begin{equation}
\left\{\begin{array}{ll} \sigma=f_1(e_1)\\
e=(1-n)e_1+nf_2^{-1}[f_1(e_1)]\,,
 \end{array} \right.
 \end{equation}
which can be viewed as the parametric equations for the engineering
stress-strain curve. We note that $n$ is actually the width (scaled
by $L$) of the localization zone in the reference configuration, as
material points in region B are in the localization zone in the
post-peak region. Obviously, system (3) couples the size effect and
nonlinear effect.

Now, we differentiate system (3) with respect to $e_1$ to obtain
\begin{equation}
\left\{\begin{array}{ll} \frac{d\sigma}{de_1}=f_1^{'}(e_1)\\
\frac{de}{de_1}=(1-n)+n\frac{f_1^{'}(e_1)}{f_2^{'}(e_2)}\,.
\end{array} \right.
 \end{equation}
If $(1-n)f_2^{'}(e_2)+nf_1^{'}(e_1)\neq0$, we have
\begin{equation}
\frac{d\sigma}{de}=\frac{f_1^{'}(e_1)f_2^{'}(e_2)}{(1-n)f_2^{'}(e_2)+nf_1^{'}(e_1)}\,\,,
\end{equation}
\begin{equation}
\frac{d^2\sigma}{de^2}=\frac{n[f_1^{'}(e_1)]^3f_2^{''}(e_2)+(1-n)f_1^{''}(e_1)[f_2^{'}(e_2)]^3}{[nf_1^{'}(e_1)+(1-n)f_2^{'}(e_2)]^3}.
\end{equation}
In order to analyze the sign of (5), we define
\begin{equation}
g(e_1,e_2;n)=nf_1^{'}(e_1)+(1-n)f_2^{'}(e_2)\,,
\end{equation}
\begin{equation}
m(e_1,e_2)=\frac{f_2^{'}(e_2)}{f_2^{'}(e_2)-f_1^{'}(e_1)}\,,
\end{equation}
\begin{equation}
G(e_1,e_2)=[f_1^{'}(e_1)]^2f_2^{''}(e_2)-[f_2^{'}(e_2)]^2f_1^{''}(e_1)\,.
\end{equation}
The above three functions can be viewed as functions of either $e_1$
or $e_2$ by the relations between them as shown above. We note that
$m(e_1,e_2)$ depends on the slopes (the first-order derivatives) of
the constitutive curves, $G(e_1,e_2)$ depends on the convexities
(the second-order derivatives) of the constitutive curves and
$g(e_1,e_2;n)$ depends on the size parameter $n$. We also point out
that $m(e_1,e_2)=n$ is equivalent to $g(e_1,e_2;n)=0$. We shall see
that $n$ has an important influence on the structural response.
\section{Post-peak Curves and Conditions}
Assuming that $f_2^{'}(e_{20})=0$ and $f_2^{''}(e_{20})<0$, that is
to say $e_{20}$ is a local maximum of $f_2$. Then, for different
$f_1$, $f_2$ and $n$ the four cases shown in Figure 2 can arise.
Next, we shall establish the conditions for each case.
\begin{figure}
 \begin{center} 
  \includegraphics[width=8.5cm]{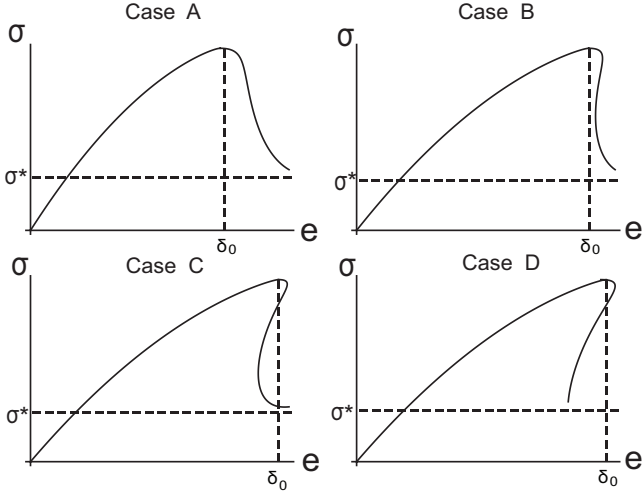}\\
  \caption{Four cases of the post-peak engineering stress-strain curves}\label{figure 2}
\end{center}
\end{figure}
\subsection{Case A: Stable Softening}
For the structure to be in stable softening (i.e., $d\sigma/de<0$),
from (5), it is easy to see the necessary and sufficient condition
is
\begin{equation}
g(e_1,e_2;n)=nf_1^{'}(e_1)+(1-n)f_2^{'}(e_2)>0, \textmd{for}\,
e_1\in [e_{11},e_{10}].
\end{equation}
From which, we get
\begin{equation}
\frac{b}{L}=n>n_0:=\max\limits_{e_1\in
[e_{11},\,e_{10}]}m(e_1,e_2)=\max\limits_{e_2\in
[e_{20},\,e_{21}]}m(e_1,e_2)\,\,.
\end{equation}
\subsection{Cases B and C: Snap-Through}
Now, we focus on the interval $n\in(0,n_0]$. There are several
possibilities, as shown in Figure 2. Before analyzing the remaining
cases, we point out that the initial part (i.e., the part close to
the peak) of the post-peak curve is in a state of stable softening
for the conditions imposed on $f_1$ and $f_2$. In fact,
$g(e_{10},e_{20})=nf_1^{'}(e_{10})>0$ , and at the peak point, we
have $d\sigma/de<0$. By continuity, there must be a part of the
post-peak curve for $e$ close to $\delta_0$
($\delta_0=(1-n)e_{10}+ne_{20}$) in which $d\sigma/de<0$. Also, at
$e_1=e_{10},e_2=e_{20}$, we have
$d^2\sigma/de^2=f_2^{''}(e_{20})/n^2<0.$ This would be useful for
our later derivation.

We see that each of Case B and Case C represents a snap- through
case. Here, snap-through is defined to be the point at which the
slope of the force-displacement curve becomes infinite. As a result,
when displacement (elongation) crosses this point, the force may
experience a sudden drop. Firstly, let us consider the similarities
between Case B and Case C. There are two turning points (the points
at which $d\sigma/de=\infty$) in both curves. From (5), it can be
seen that this is equivalent to that the equation
\begin{equation}
g(e_1,e_2;n)=nf_1^{'}(e_1)+(1-n)f_2^{'}(e_2)=0
\end{equation}
has two roots, say $e_{11}^{*}$ and $e_{12}^{*}$
\,($e_{11}^{*}>e_{12}^{*}$). The following theorem provides a
necessary and sufficient condition for the occurrence of the two
turning points.
\begin{thm}  If two turning points occur, then the
function $G(e_1,e_2)$ must change sign at least once for $e_1\in
[e_{11},e_{10}]$. On the other hand, if the sign of $G(e_1,e_2)$
changes only once for $e_1\in [e_{11},e_{10}]$, then for any $n\in
[n_1,n_0]$, two turning points occur, where $n_1=m(e_{11},e_{21})$.
\end{thm}
\begin{proof} If two turnings occur, then the sign of the function
$g(e_1,e_2)$ changes twice (cf. Case B or Case C in Figure 2). So,
we get
\begin{equation}
g(e_{11},e_{21};n)=nf_1^{'}(e_{11})+(1-n)f_2^{'}(e_{21})>0\,.
\end{equation}
Thus,
\begin{equation}
\frac{f_2^{'}(e_{21})}{f_1^{'}(e_{11})}> -\frac{n}{1-n}\,.
\end{equation}
Since
\begin{equation}
-\frac{n}{1-n}\geq-\frac{n_0}{1-n_0}\,,
\end{equation}
\begin{equation}-\frac{n_0}{1-n_0}=\min\limits_{e_1\in
[e_{11},\,e_{10}]}\frac{f_2^{'}(e_2)}{f_1^{'}(e_1)}\,,
\end{equation}
we have
\begin{equation}\frac{f_2^{'}(e_{21})}{f_1^{'}(e_{11})}>\min\limits_{e_1\in
[e_{11},\,e_{10}]}\frac{f_2^{'}(e_2)}{f_1^{'}(e_1)}.
\end{equation}
Suppose that the minimum is attained at $e_{12}$ (the corresponding
$e_2$ is given by $f_2^{-1}[f_1(e_{12})]=e_{22}$). That is,
$f_2^{'}(e_{22})/f_1^{'}(e_{12})=\min\limits_{e_1\in
[e_{11},\,e_{10}]}\{f_2^{'}(e_2)/f_1^{'}(e_1)\}$. If we view
$f_2^{'}(e_2)/f_1^{'}(e_1)$ as a function of $e_2$, then according
to the Lagrange Mean Value theorem, there exists an $\xi\in
(e_{22},e_{21})$ such that
\begin{equation}
\frac{d}{de_2}[\frac{f_2^{'}(e_2)}{f_1^{'}(e_1)}]|_{e_2=\xi}=
\frac{\frac{f_2^{'}(e_{21})}{f_1^{'}(e_{11})}-\frac{f_2^{'}(e_{22})}{f_1^{'}(e_{12})}}{e_{21}-e_{22}}>0\,.
\end{equation}
As
\begin{eqnarray}
\frac{d}{de_2}[\frac{f_2^{'}(e_2)}{f_1^{'}(e_1)}]&=&\frac{[f_1^{'}(e_1)]^2f_2^{''}(e_2)-[f_2^{'}(e_2)]^2f_1^{''}(e_1)}{[f_1^{'}(e_1)]^3}\,,
\end{eqnarray}
we have
\begin{equation}
[f_1^{'}(\zeta)]^2f_2^{''}(\xi)-[f_2^{'}(\xi)]^2f_1^{''}(\zeta)>0,\,
 \textmd{where}\,\,\,\zeta=f_1^{-1}[f_2(\xi)]\,,
\end{equation}
which implies that $G(\zeta,\xi)>0$. Since
$G(e_{10},e_{20})=[f_1^{'}(e_{10})]^2f_2^{''}(e_{20})<0$, \,the sign
of $G(e_1,e_2)$ changes for $e_1\in [e_{11},e_{10}]$.

On the other hand, suppose that for $e_1\in [e_{11},e_{10}]$,
$G(e_1,e_2)$ changes sign once. Now, we consider the function
$m(e_1,e_2)$ (cf., (8); we regard it as a function of $e_2$). We now
show that for any $n\in[n_1,n_0]$, equation (12) has two roots. In
fact, it is easy to get
\begin{equation}
\frac{dm}{de_2}=\frac{-G(e_1,e_2)}{[f_2^{'}(e_2)-f_1^{'}(e_1)]^2f_1^{'}(e_1)}\,\,.
\end{equation}
Thus, $dm/de_2$ also changes sign once. We also note that $n_0$ is a
maximum of $m(e_1,e_2)$, say, attained at $e_{2n}$. Then,
$dm/de_2=0$ at $e_2=e_{2n}$. On the other hand,
\begin{equation}\frac{dm}{de_2}|_{e_2=e_{20}}=\frac{-G(e_{10},e_{20})}{[f_2^{'}(e_{20})-f_1^{'}(e_{10})]^2f_1^{'}(e_{10})}>0.
\end{equation}
So, the curve $m(e_1,e_2)$ should have the characteristics shown in
Figure 3.
\begin{figure}
 \begin{center} 
  \includegraphics[width=8cm]{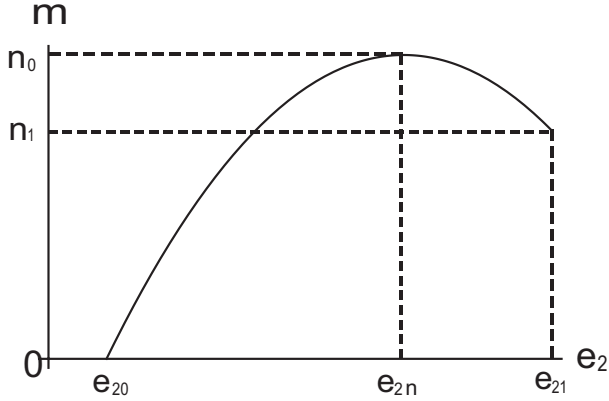}\\
  \caption{Schematic illustration of $m(e_1,e_2)$}\label{figure 2}
\end{center}
\end{figure}
Thus, for any $n\in [n_1,n_0]$, $n=m(e_1,e_2)$ has two roots, which
then implies that $g(e_1,e_2;n)$ has two zeros. This completes the
proof of the second part of the theorem.
\end{proof}
If $f_1$ is linear in the post-peak region, then
$G(e_1,e_2)=[f_1^{'}(e_1)]^2f_2^{''}(e_2)$. Consequently, the sign
of $G(e_1,e_2)$ depends on the convexity of $f_2$. We have the
following corollary:
\begin{cor} For $f_1$ being linear in the post-peak region, if two turning points occur, then the convexity of $f_2$
must change at least once for $e_1\in [e_{11},e_{10}]$. On the other
hand, if the convexity of $f_2$ changes once, then for any $n\in
[n_1,n_0]$, two turning points occur in the post-peak curve.
\end{cor}
\begin{rem}
Usually, $f_{12}^{''}(e_1)$ should be small ($f_{12}^{''}(e_1)=0$
for $f_1$ being linear). Thus, the sign of $G(e_1,e_2)$ is primarily
determined by the sign of $f_{22}^{''}(e_2)$. So, one may say that a
necessary condition for the snap-through (i.e., there are two
turning points in the post-peak curve) is the change of the
convexity of the constitutive curve of the softening part.
\end{rem}
Now, let us consider the differences between Case B and Case C.
Recall that $e_{11}^{*}$ and $e_{12}^{*}$($e_{11}^{*}>e_{12}^{*}$)
are the two roots of equation (13). For Case B, we have
\begin{equation}
\delta_2^{*}=(1-n)e_{12}^{*}+nf_2^{-1}[f_1(e_{12}^{*})]>\delta_0\,\,.
\end{equation}
 While for Case C, we have
\begin{equation}
\delta_2^{*}=(1-n)e_{12}^{*}+nf_2^{-1}[f_1(e_{12}^{*})]\leq\delta_0\,.
\end{equation}
In other words, in Case C the post-peak curve has entered the
pre-peak region, while in Case B it has not. We find that, for given
$f_1$ and $f_2$, there are some conditions for the occurrence of
Case C. For simplicity, we assume that the sign of $G(e_1,e_2)$
changes once, say,
\begin{equation}
\left\{\begin{array}{ll} G(e_1,e_2)<0,\,\,e_2\in [e_{20},e_{22})\,,\\
G(e_{12},e_{22})=0\,,\\
G(e_1,e_2)>0,e_2\in (e_{22},e_{21}]\,.
\end{array} \right.
 \end{equation}
The following theorem provides a necessary and sufficient condition
for Case C.
\begin{thm} Under assumption (25) and
$n\in[n_1,n_0]$, a necessary and sufficient condition for the
occurrence of Case C is
\begin{equation}
f_2^{'}(e^{*}_{22})(e^{*}_{22}-e_{20})\geq
f_1^{'}(e^{*}_{12})(e^{*}_{12}-e_{10})\,.
\end{equation}
\end{thm}
\begin{proof}
From (24), we have
\begin{equation}
(1-n)e^{*}_{12}+ne^{*}_{22}\leq
(1-n)e_{10}+ne_{20},\,\textmd{where}\,\,\,
e^{*}_{22}=f_2^{-1}[f_1(e_{12}^{*})].
\end{equation}
On the other hand,
\begin{equation}
g(e^{*}_{12},e^{*}_{22};n)=nf_1^{'}(e^{*}_{12})+(1-n)f_2^{'}(e^{*}_{22})=0\,.
\end{equation}
From the above two equations, we can get (26) immediately. Also,
from (26) and (28) one can immediately deduce (27). This completes
the proof.
\end{proof}
 Assumption (25) can be made even more complicated, in that case we
 may draw the fairly complicated post-peak curves in
 \cite{Rots1989}, which were obtained by numerical methods. It should be pointed out
that inequality (26) is another requirement among $f_1$, $f_2$ and
$n$. For given $f_1$ and $f_2$, it provides another bound (say
$n_2$) for $n$, since $e^*_{22}$ and $e^*_{12}$ are related to $n$
as equation (28) shows.

For $f_1$ being linear in the post-peak region, it is easy to show
that (26) becomes
\begin{equation}
f_2^{'}(e^{*}_{22})\geq
\frac{f_2(e^{*}_{22})-f_2(e_{20})}{e^{*}_{22}-e_{20}}.
\end{equation}
This implies that for the constitutive relation $\sigma=f_2(e_2)$,
the secant line joining the point $e^*_{22}$ and the peak $e_{20}$
should be steeper than the tangent line at $e^*_{22}$ (see Figure
4). Combined with Corollary 3.1, we have the following corollary:
\begin{cor}
For $f_1$ being a linear function, if the convexity of $f_2$ changes
once, then a necessary and sufficient condition for the occurrence
of Case C is $n\in [n_1,n_0]$, and inequality (29) holds.
\end{cor}
\begin{figure}
 \begin{center} 
  \includegraphics[width=8cm]{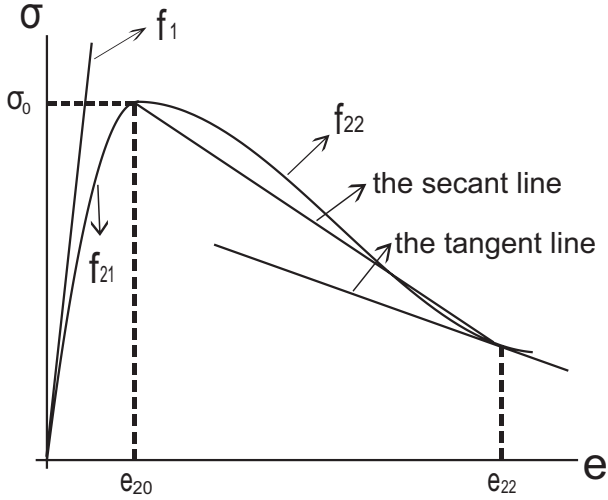}\\
  \caption{Diagrammatic representation of inequality (29) }\label{figure 4}
\end{center}
\end{figure}
\subsection{Case D: Snap-back}
In Case D, there is a snap-back in the structural response. Here, we
say that snap-back occurs when the slop of the force-displacement
curve becomes positive and remains positive in the post-peak
response. Obviously, in this case there is only one turning point
(see Figure 2). The following theorem provides a critical $n$ for
the occurrence of Case D.
\begin{thm} If $G(e_1,e_2)$ changes sign once (cf.(25))
or does not change sign for $e_1\in [e_{11},e_{10}]$, then for Case
D to occur, a necessary and sufficient condition is $n<n_1$.
\end{thm}
\begin{proof}
First, suppose that $G(e_1,e_2)$ changes sign once. Then
$m(e_1,e_2)$ has the characteristics shown in Figure 3. It is
obvious that a necessary and sufficient condition for the occurrence
of Case D is that there exists only one root for equation (12).
While from Figure 3, it is easy to see that a necessary and
sufficient condition is $n<n_1$. Second, suppose that $G(e_1,e_2)$
does not change sign. Since $dm/de_2>0$ (as $dm/de_2>0$ at
$e_2=e_{20}$), we have
\begin{equation}n_0=\max\limits_{e_2\in
[e_{20},\,e_{21}]}m(e_1,e_2)=m(e_{11},e_{21})=n_1.
\end{equation}
 Obviously, a
necessary and sufficient condition for $n=m(e_1,e_2)$ to have one
and only one root (i.e., $g(e_1,e_2;n)$ has one and only one zero)
is $n<n_1$. Thus we complete the proof.
\end{proof}
\begin{rem}
In this section, we derive some requirements on the constitutive
functions, together with three critical values ($n_0$, $n_1$, $n_2$)
of the size parameter.  Providing the constitutive requirements are
met, the structural response may have different behaviors for $n$ in
different intervals according to the above critical values. Thus,
the results also show the important size effects.
\end{rem}
\section{Illustrative Examples}
In this section, we give two examples to illustrate the theoretical
results obtained in Section 3. The following two examples can be
referred to as two different physical processes. One can easily
check that the functions in the following examples satisfy the
conditions we have proposed, in particular, (25).

Example 1: Consider the following constitutive relations:
$$f_{11}(e_1)=\csc\frac{9\pi}{20}\sin
  (50\pi e_1),$$
  $$f_{12}(e_1)=(50\pi\csc\frac{9\pi}{20})e_1-\frac{9\pi}{20}\csc\frac{9\pi}{20}+1,$$
   $$f_{21}(e_2)=\sin (50\pi e_2),$$
$$f_{22}(e_2)=65625[\frac{1}{3}(e_2-0.03)^3-0.0004(e_2-0.03)]+0.65.$$
   Here $\sigma_{0}=1$, and we take $\sigma^\star=0.301$. The details are shown in Figure 5(a).

    We have taken $f_{12}$ being a linear function, which represents
    the physical situation that at the peak region A of the structure
    has entered the plastic state. We find the critical values of $n$ based on the theoretical analysis
in Section 3:
    $n_0=0.142$, $n_1=0.0158$, $n_2=0.110$ (a bound for $n$ found from inequality (29)).
    Specifically, Case A occurs if $n>0.142$; Case B occurs if $0.110<n\leq
    0.142$; Case C occurs if $0.0158\leq n\leq0.110$; Case D occurs
    if $n<0.0158$. Accordingly, by taking $n$ to be in different intervals, we get
the four cases as we have predicted in Section 3. They are shown in
Figure 5(b).

To reflect the size effect on the localization zone, curves of the
width of the localization zone in the current configuration versus
the total elongation are shown in Figure 6. This width is denoted by
$d$, whose expression is given by $d=n(1+e_2)$. Here, for the
purpose of clearness, we have used different scales for different
curves. It can be seen that this width increases slowly in the
pre-peak region and increases rapidly in the post-peak region. For
the stable softening case ($n=0.167$), there is only one value of
$d$ for a given $e$. For the snap-back case ($n=0.0156$), there are
two values of $d$ for a given $e$. Also, $d$ increases very fast, as
$e$ decreases in the post-peak region. For the two snap-through
cases ($n=0.1$ and $n=0.125$), there are three values of $d$ for $e$
in some intervals. Thus, $d$ may jump from a small value to a large
value for $e$ in these intervals, i.e., the localization zone may
suddenly widen. Thus, the size parameter $n$ has an important
influence on the localization zone.
\begin{figure}
   \centering
  \begin{minipage}{0.45\textwidth}
     \centering \includegraphics[width=70mm,height=48mm]{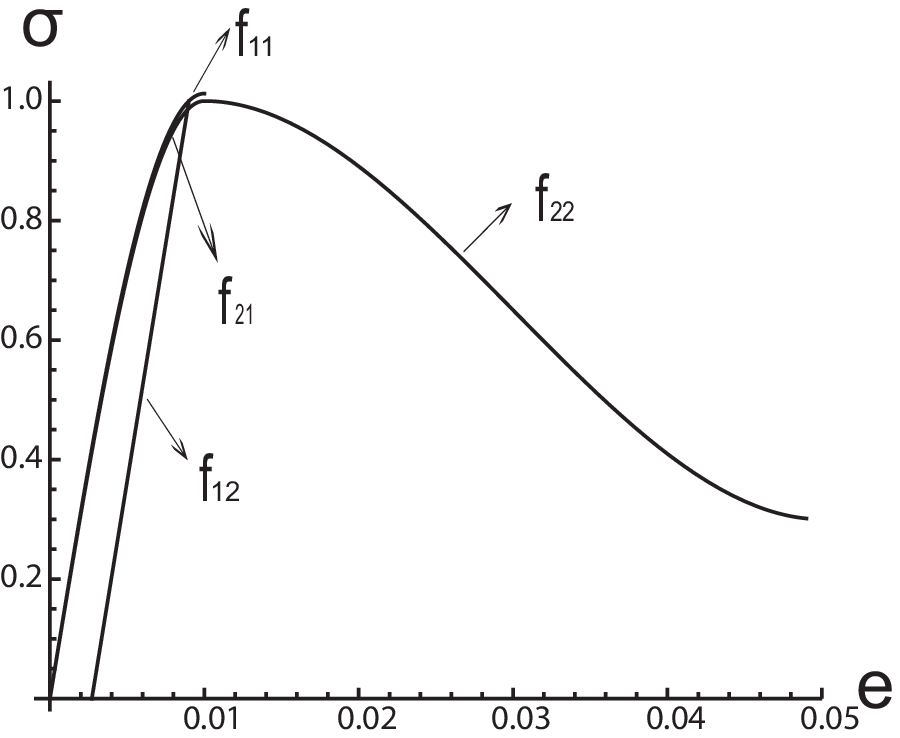}\\(a)
  \end{minipage}
  \begin{minipage}{0.45\textwidth}
     \centering \includegraphics[width=70mm,height=48mm]{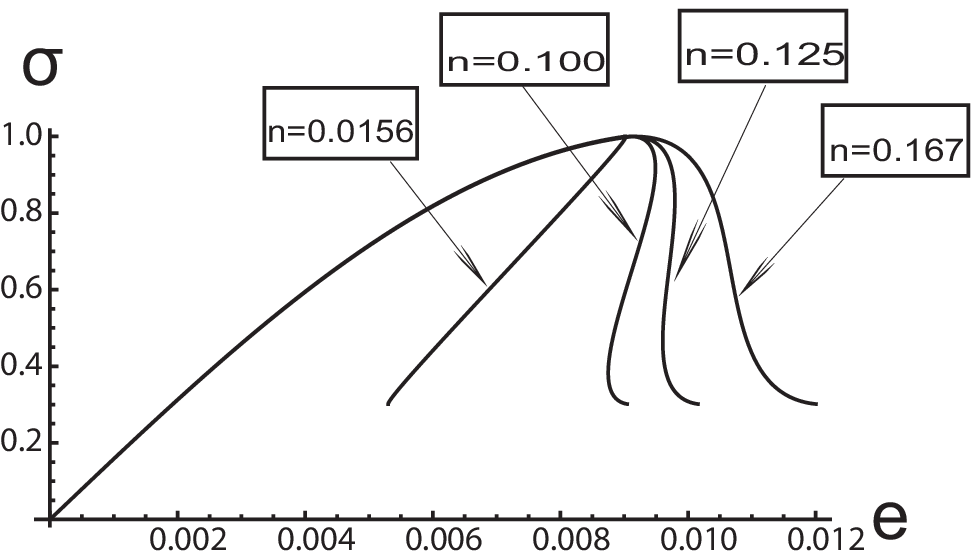}\\(b)
  \end{minipage}
  \caption{(a) The constitutive curves of Example 1;
    (b)The engineering stress-strain curves for different $n$ in Example 1. }
\end{figure}
\begin{figure}
 \begin{center} 
  \includegraphics[width=8cm]{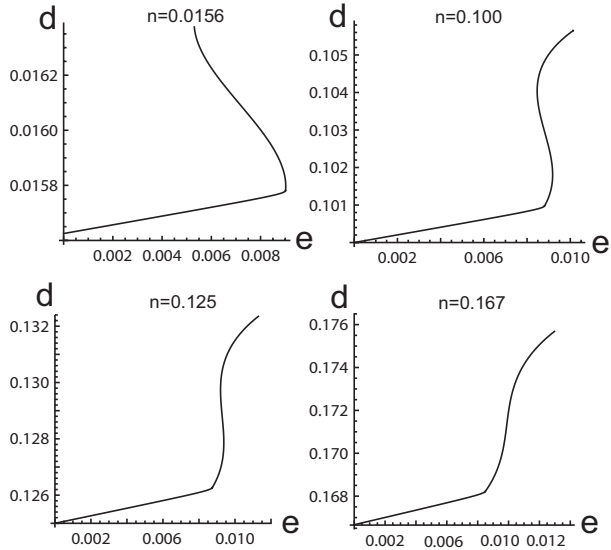}\\
  \caption{ The $d-e$ curves for different $n$ in Example 1}\label{figure 6}
\end{center}
\end{figure}

Example 2: In this example, region A is assumed to be in nonlinear
elasticity (loading or unloading), so the constitutive functions
$f_{11}$ and $f_{12}$ are the same. The constitutive relations are
listed below.
$$f_{11}(e_1)=f_{12}(e_1)=-11000(e_1-0.01)^2+1.1,$$
$$f_{21}(e_2)=-\frac{10^6}{81}(e_2-0.009)^2+1,$$
$$f_{22}(e_2)=65625[\frac{1}{3}(e_2-0.029)^3-0.0004(e_2-0.029)]+0.65.$$
Here $\sigma_0=1$, and we take $\sigma^\star=0.333$. The details are
shown in Figure 7(a). Critical values of $n$ are: $n_0=0.174$,
$n_1=0.0614$, $n_2=0.136$ (a bound for $n$ found from inequality
(26)). The intervals for different cases are: Case A occurs if
$n>0.174$; Case B occurs if $0.136<n\leq 0.174$; Case C occurs if
$0.0614\leq n\leq 0.136$; Case D occurs if $n<0.0614$. The curves
for $n$ taking four values in these four different intervals are
shown in Figure 7(b), which agree with our theoretical predictions
in Section 3. Curves of the width of the localization zone in the
current configuration versus the total elongation are also shown
(see Figure 8). Once again, from these curves, one can see the
important influence of the size parameter $n$ on the localization
zone.
\begin{figure}
   \centering
  \begin{minipage}{0.45\textwidth}
     \centering \includegraphics[width=70mm,height=45mm]{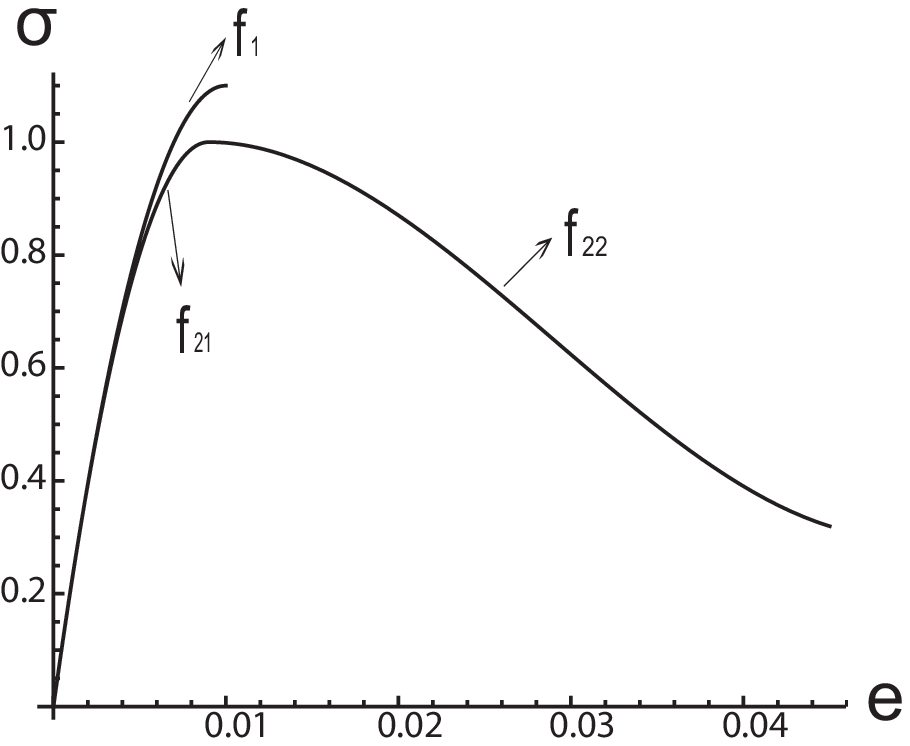}\\(a)
  \end{minipage}
  \begin{minipage}{0.45\textwidth}
     \centering \includegraphics[width=70mm,height=45mm]{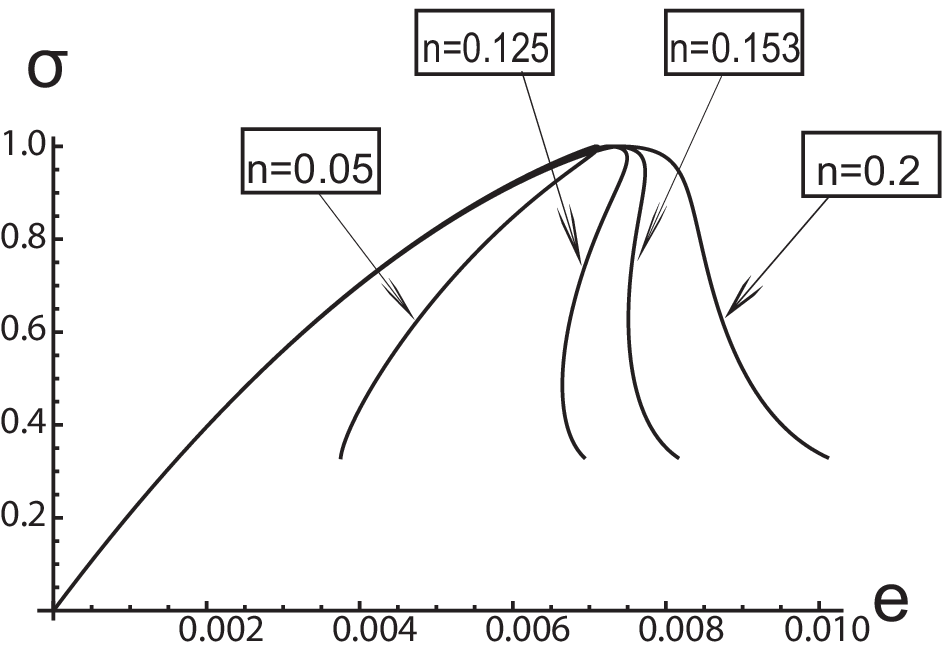}\\(b)
  \end{minipage}
  \caption{(a)
The constitutive curves of Example 2; (b) The engineering
stress-strain curves for different $n$ in Example 2.}
\end{figure}
\begin{figure}
 \begin{center} 
  \includegraphics[width=8cm]{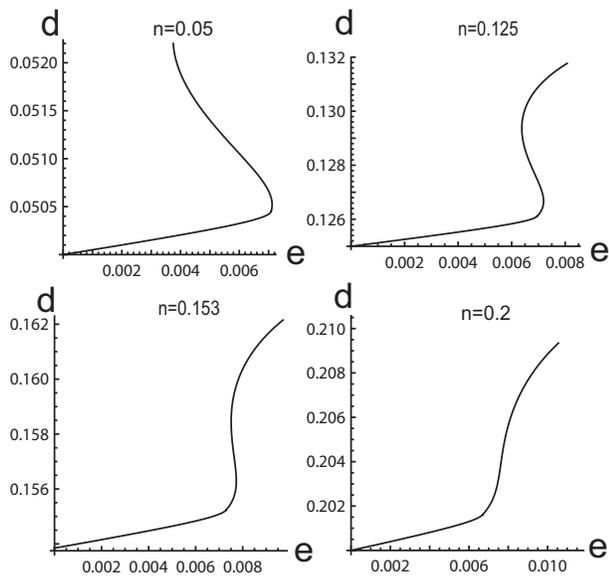}\\
  \caption{ The $d-e$ curves for different $n$ in Example 2}\label{figure 8}
\end{center}
\end{figure}
\section{Concluding Remarks and Future Tasks}
An analytical study is performed on the post-peak structural
response of strain-softening with localization. In a general
nonlinear setting, after taking standard mathematical analysis to
the parametric equations, we manage to handle the nonlinear and size
effects. Qualitative requirements on the constitutive functions and
quantitative requirements on the size effect are derived, especially
for the snap-through phenomenon. The results are consistent with
earlier experimental and computational results. It seems that the
four cases studied analytically here are quite representative. The
theoretical results may be of value for the verification of
computational algorithms and can shed some light on the mechanisms
of instabilities associated with strain-softening. Especially, we
have shown that the convexity change is a necessary condition for
the snap-through phenomenon. As softening with localization is
important for understanding the failure evolution in structures,
future work will focus on considering structures with different
configurations.
\section*{Acknowledgement}
The work described in this paper is supported by a grant from City
University of Hong Kong (Project No. 7002366), the National Natural
Science Foundation of China (Nos. 10721062, 90715037, 10728205 and
10902021), the Program for Changjiang Scholars and Innovative
Research Team in University of China (PCSIRT), the 111 Project (No.
B08014) and the National Key Basic Research Special Foundation of
China (No. 2010CB832704).






\end{document}